\documentclass[12pt,reqno,a4paper]{amsart}

\usepackage{lipsum,subeqnarray}
\usepackage{amsfonts}
\usepackage{graphicx}
\usepackage{amsmath}
\usepackage{mathtools}
\usepackage{amssymb}
\usepackage{enumerate}
\usepackage{slashed}
\usepackage{graphicx}
\usepackage{newlfont}
\usepackage{amsrefs}
\usepackage{comment}
\usepackage[normalem]{ulem}
\usepackage{mathtools}
\usepackage{accents}
\usepackage{leftidx}
\usepackage{bbm}
\usepackage{amsthm}
\usepackage{appendix}
\usepackage{comment}
\usepackage{subfig}
\usepackage{float}
\usepackage{subfloat}
\usepackage{url}

\numberwithin{equation}{section}

\theoremstyle{plain}
\newtheorem{theorem}{Theorem}[section]

\newtheorem{proposition}[theorem]{Proposition}
\newtheorem*{theorem*}{Conjecture}

\usepackage{xcolor}

\theoremstyle{definition}

\theoremstyle{remark}

\usepackage{stackengine}
\usepackage{subeqnarray} 
\usepackage{esint}

\begin{document}
\title[Critical String Theory in a $D=4$ Robertson-Walker Background]{Critical String Theory in a $D=4$ Robertson-Walker Background and the Large Scale Structure of Spacetime}
 
\author{Jaroslaw S. Jaracz}

\begin{abstract}
We show that $4$-dimensional Robertson-Walker spacetimes can be constructed for which all of the beta functions vanish to leading order, yielding consistent string theory without the need for any extra dimensions. We find that, up to scaling of some constants, there is a unique static solution, which we refer to as an anti-Einstein static universe. This is an infinite universe with negative scalar curvature. The associated stress energy tensor can be interpreted as a perfect fluid with a negative energy density. Interestingly, a fluid with a negative energy density was proposed in \cite{farnes2018unifying} as an ad hoc hypothesis to serve as a possible explanation for dark energy and dark matter. Here, such a fluid spontaneously appears by trying to fit string theory into only $4$ dimensions. Since our universe appears approximately static, we can look at perturbations away from the anti-Einstein static universe. This has to be done numerically and we only do it for a few choices of initial conditions (a full systematic numerical analysis of solutions will have to wait for the future). We find that these solutions are very sensitive to the initial conditions and yield a variety of behaviors. We hope that this behavior is rich enough to match cosmological observations by appropriately choosing the initial conditions (which needs to be systematically investigated with a search in the space of initial conditions). One of these solutions that we found is particularly interesting. It has a hubble parameter which is negative in the past (meaning a contracting universe), then a point at which the hubble parameter changes sign, so the universe starts expanding which could be identified with a big bang, after which the hubble parameter continues growing. Throughout all of this time, the acceleration is positive since the Hubble parameter is increasing. Finally, we comment that our result does not contradict \cite{antoniadis1988cosmological} as it seems that the authors overlooked the possibility of a purely complex axion field which allows for the Robertson-Walker metrics to make a contribution $c_{RW}\geq 4$ to the central charge. Thus, we can interpret dark matter and dark energy as being parts of a mechanism needed to keep $4$-dimensional string theory consistent. 
\end{abstract}
\maketitle

\section{Introduction}

Before starting the article proper, we need to address the (rightfully so) highly skeptical readers. As we will see the conclusions of the article, follow from some straightforward calculations, and then looking at the perturbations of the only static solution to the resulting equations. Then naturally, if it is so simple, how come no one has done it before? We believe the answer is two fold. First, the idea of using a $4$-dimensional Robertson-Walker spacetime for a background in string theory is very natural once one has learned about the beta functions. But the second instinct is that if this worked, since string theory has been around for so long someone would have already done it, which serves to filter out a large number of people interested in the idea. Second, the next instinct would be to check the literature, where one of the main references is \cite{antoniadis1988cosmological} where the authors plainly state that the contribution to the central charge made by the Robertson-Walker metrics is $c_{RW}<4$ meaning that $D=4$ will not work, and so the case is closed. This filters out the remaining people interested in the problem.

However, we were not aware of \cite{antoniadis1988cosmological}, so we simply tried to solve the beta functions for $D=4$, found this to be quite easy, and then were confused why we had not heard of this before. This brought us to \cite{antoniadis1988cosmological}, which claims that $c_{RW}<4$. So why is this not a contradiction of our result? The key is found in the paragraph prior to equation (12) where the authors obtain the forumla
\begin{align}
    K=\frac{1}{4}\beta^2 e^{-2\phi_0}
\end{align}
for the static case where $\beta$ is a constant related to the axion field. So here the authors concluded that $K\geq0$ in which case indeed $c_{RW}<4$. We had arrived at the same formula \eqref{CriticalEquation}, since we have $\phi=2\Phi$ comparing the notations of the two articles, but we instead concluded that $\beta=i|\beta|$ needs to be purely complex, and $K<0$. This allows us to make the contribution to the central charge be whatever we want. Why didn't the authors consider a complex axion? The article was written before the discoveries of the dark energy and dark matter problems. At the time, everyone "knew" all matter has nonnegative energy density, and thus the scalar curvature should be nonnegative. We suppose that by the time the dark energy and dark matter problems were discovered, the string community had moved onto compactifications, branes, and eventually Ads/CFT so no one thought to go back to $4$-dimensional string theory. 

 Finally, we mention if one accepts the calculations in \cite{antoniadis1988cosmological}, one can skip the majority of our paper and simply let the axion field be purely complex in that article. However, since we believe this is an important result, we include all of the calculations in painstaking detail. We believe the equations are right, as they also match the equations of \cite{antoniadis1988cosmological}, but if there is a mistake, it should be easy to spot the mistake. Also, our calculations are done in a different way, as we treat $\beta^\Phi=0$ as a conserved quantity dependent on initial conditions, rather than an equation of motion. In both cases, one arrives at $3$ equations, one of which is dependednt, but we think our way is better for the numerical analysis of solutions. Now that we have addressed this issue, we can move onto the article. 

String theory is one of the leading candidates for a theory of quantum gravity. It was originally proposed in the 1960s as a way of explaining the strong interaction. It was found that a massless spin-$2$ particle naturally arises on its own in the theory, and eventually it was realized that this particle could be identified with the graviton. Thus, string theory is a quantum theory where gravity pops out on its own, without being forced into the theory. See \cite{schwarz2012early} for a sketch of the early history. 

However, like all theories of quantum gravity, string theory has its problems. First, in the bosonic case, the dimension of the spacetime has to be $26$. By the addition of supersymmetry \cite{martin1998supersymmetry}, this can be lowered down to $10$ or $11$ depending on the type of string theory used. These extra dimensions have never been observed, so they are usually compactified, which leads to the study of Calabi-Yau manifolds, which then leads to the so-called landscape problem, where there are roughly $10^{500}$ different compactifications, and no way to choose which one is the correct one to describe our universe. See \cite{susskind2003anthropic} and \cite{Read2021-REATLA-2} for some perspectives. Moreover, supersymmetry has never been observed, even though evidence for it was expected to be found at the LHC, see for example \cite{khachatryan2014searches}. These are all valid criticisms. 

It is known that the dimension in which string theory is consistent can be changed by the introduction of dilaton and axion fields. If the dimension could be decreased to $4$, this would solve many problems. However since the universe is approximately homogeneous and isotropic, these fields should have very small spatial gradients, and the solution should be an approximately Robertson-Walker spacetime.

Next, we review some basics facts from string theory and cosmology which we will need. In bosonic string theory in a curved background the Polyakov action is given by
\begin{align} \label{PolyakovAction}
    \begin{split}
        S_P= &\frac{1}{4\pi \ell_S^2} \int d^2 \xi \left[ \sqrt{h} h^{\alpha \beta} g_{\mu \nu }(X) + \epsilon^{\alpha \beta} B_{\mu \nu}(X)     \right] \partial_\alpha X^\mu \partial_\beta X^\nu \\ &+ \frac{1}{4\pi} \int d^2 \xi \sqrt{h}R^{(2)} \Phi(X)
    \end{split}
\end{align}
where $h_{\alpha \beta}$ is the worldsheet metric (of signature $(++)$ after a Wick rotation), $g_{\mu \nu}$ is the background metric of "mostly plus" signature, $B_{\mu \nu}$ is an antisymmetric tensor, $\Phi$ is the dilaton field and $\ell_s$ is the string length parameter, related to the usual Regge slope $\alpha'$ by
\begin{align}
    \ell_s=\sqrt{\alpha'}.
\end{align}
The beta functions for this action are well known and are
\begin{align}
    \label{BetaFunctiong}
    \frac{\beta_{\mu \nu}^g}{\ell_s^2} = R_{\mu \nu} + 2\nabla_\mu \nabla_\nu \Phi -\frac{1}{4}H_{\mu \rho \sigma}{H_{\nu}}^{\rho \sigma} + O(\ell_s^2)
\end{align}
\begin{align}
    \label{BetaFunctionB} 
    \frac{\beta_{\mu \nu}^B}{\ell_s^2} = -\frac{1}{2}\nabla^\rho \left( e^{-2\Phi} H_{\mu \nu \rho}   \right) + O(\ell_s^2)
\end{align}
\begin{align}
    \label{BetaFunctionPhi} 
    \beta^\Phi=D-26 +\frac{3}{2}\ell_s^2 \left[ 4\nabla^\mu \Phi \nabla_\mu \Phi -4 \square \Phi -R + \frac{1}{12}H^2\right] + O(\ell_s^4)
\end{align}
where
\begin{align}
    H_{\mu \nu \rho} = \partial_\mu B_{\nu \rho} + \partial_\nu B_{\rho \mu} + \partial_\rho B_{\mu \nu}
\end{align}
and $\square =\nabla^\mu \nabla_\mu$ is the d'Alembertian. To obtain a consistent theory, the beta functions have to vanish, at least to leading order. See Section 6.1 in \cite{Kiritsis:2019npv} or Section 3.7 in \cite{Polchinski:1998rq}. 

One important thing to note is that the metric in the above formulas is the metric in the string frame. It is related to the metric in the Einstein frame $g_{\mu \nu}^E$ by 
\begin{align}
    g^E_{\mu \nu} = e^{-\frac{4\Phi}{D-2}} g_{\mu \nu}
\end{align}
so in the case where $D=4$ this gives,
\begin{align}\label{RelationEinsteinToStringMetric}
    g^E_{\mu \nu} = e^{-2\Phi} g_{\mu \nu}.
\end{align}
The metric in the Einstein frame is considered the "physical metric." This is because the vanishing of the beta functions can also be obtained from the variation of an action which in the Einstein frame takes the form
\begin{align}
    S_{E}=\frac{1}{2\kappa^2}\int d^Dx \sqrt{g^E}\left[ R - \frac{4}{D-2}\left( \nabla \Phi  \right)^2 - \frac{e^{-8\Phi/(D-2)}}{12}H^2 + 2 \frac{26-D}{3\ell_s^2} e^{4\Phi/(D-2)} 
    \right] + O(\ell_s^2)
\end{align}
and looks like the ordinary Einstein-Hilbert action where $\Phi$ and $H$ play the role of matter fields. Thus comparing this with the usual Einstein-Hilbert action suggests that we view $g^E$ as the usually observed, that is physical, metric.

Next, we need to review some facts from general relativity. The Robertson-Walker spacetimes are the homogenous, isotropic solutions to Einstein's equations
\begin{align} \label{EinsteinEquations}
    G_{\mu \nu} = 8\pi T_{\mu \nu}.
\end{align}
Assuming simple connectedness and dimension $4$ for the spacetime, the manifold $M$ and the metric $g$ must be of the form
\begin{align} \label{RWManifoldMetric} \begin{split}
    M=& \mathbb{R}\times \Sigma \\
    g=& -dt^2 + a^2(t) d\Sigma^2 \end{split}
\end{align}
where $\Sigma$ is either hyperbolic space, Euclidean space or the round sphere all in dimension $3$ and $d\Sigma^2$ is the associated metric. 

A perfect fluid stress-tensor is given by 
\begin{align} \label{GeneralPerfectFluidStressTensor}
    T_{\mu \nu} = \rho u_{\mu} u_{\nu} + P (g_{\mu \nu} + u_\mu u_\nu)
\end{align}
where $\rho$ is the fluid energy density, $P$ is the pressure, and $u^\mu$ is the $4$-velocity of the fluid. Assuming $T_{\mu \nu}$ is of the form \eqref{GeneralPerfectFluidStressTensor} with $u^\mu=\partial_t$ is the $4$-velocity of the isotropic observers, then the evolution equations for the scaling factor $a(t)$ are given by
\begin{align}
  \label{RWDDotEquation}  3\ddot{a} / a = -4\pi(\rho +3P) \\
    \label{RWDotEquation}3\dot{a}^2/a^2 = 8\pi \rho -3K/a^2
\end{align}
where dots denote derivatives with respect to $t$ and $K=-1, 0, 1$ for the cases of hyperbolic space, Euclidean space, and the round sphere respectively.

The vacuum Einstein equations with cosmological constant 
\begin{align}
    G_{\mu \nu} + \Lambda g_{\mu \nu} = 0
\end{align}
for can be put into a perfect fluid like form by writing
\begin{align} \begin{split}
    G_{\mu \nu} =& -\Lambda g_{\mu \nu} \\=& -\Lambda g_{\mu \nu} + \Lambda u_\mu u_ \nu - \Lambda u_\mu u_\nu\\ =& \Lambda u_\mu u_\nu - \Lambda (g_{\mu \nu} + u_\mu u_\nu) \end{split}
\end{align}
so comparing this with \eqref{EinsteinEquations} and \eqref{GeneralPerfectFluidStressTensor} shows the cosmological constant term can be written as a perfect fluid with
\begin{align} \label{DensityPressureCosmologicalConstant}
    \rho=\Lambda/8\pi, \quad P=-\Lambda/8\pi
\end{align}
Then using \eqref{RWDDotEquation} we get
\begin{align} \label{AccelerationCosmologicalConstant}
    3\ddot{a}/a=-4\pi(\rho + 3P)=-4\pi\left( \frac{\Lambda}{8\pi} - \frac{3\Lambda}{8\pi} \right)= \Lambda
\end{align}
so if the cosmological constant is positive $\Lambda>0$ then $\ddot{a}>0$ indicating an accelerating expansion. For a reference of these basic facts, see chapters 4 and 5 of \cite{Wald:1984rg}.

\section{Vanishing of the Beta Functions}\label{SecVanishingBetaFunctions}

First we compute certain important quantities for $g^E$ and $g$ which we will subsequently need. Then, we solve $\beta^B_{\mu \nu}=0$ to leading order since that beta function is the easiest and we need it determine $H_{\mu \nu \rho}$ for later use. We then see what equations have to be satisfied to make \eqref{BetaFunctiong} vanish. Finally, we use the well known fact that $\beta^g_{\mu \nu}= \beta^B_{\mu \nu}=0$ implies $\beta^\Phi=constant$ (see for example the paragraph after equation (6.1.8) in \cite{Kiritsis:2019npv}) to turn the equation $\beta^\Phi = 0$ into a condition on the initial values of the relevant fields which acts as a constant of motion,

\subsection{Some Basic Calculations}

As mentioned, the metric in the Einstein frame is to be thought of as the physical metric. Due to homogeneity and isotropy we require $g^E$ is a $4$-dimensional Robertson-Walker metric, that is
\begin{align}  \begin{split}
    M=& \mathbb{R}\times \Sigma \\
    g^E=& -dt^2 + a^2(t) d\Sigma^2. \end{split}
\end{align}
All objects related to this metric will be denoted with $E$, for example $\nabla^E, R^E_{\mu \nu}, \square^E,$ etc. First we collect some calculations related to this metric. We let Greek letters $\alpha, \beta, \mu, \nu,$ etc. denote spacetime indices and Latin letters $i, j, k,$ etc. denote space indices.   

Let $\overline{g}_{ij}$ denote the metric on the $3$-dimensional space slices. For all functions we let $\dot{F}=dF/dt$. We compute the Christoffel symbols. We have
\begin{align}
    g^E_{\mu \nu} =\begin{pmatrix}
     -1 & 0 & 0  & 0 \\
     0 \\ 
    0 & & a^2 \overline{g}_{ij}\\ 
    0  &
    \end{pmatrix}, \quad \quad  (g^E)^{\mu \nu} =\begin{pmatrix}
     -1 & 0 & 0  & 0 \\
     0 \\ 
    0 & & a^{-2} \overline{g}^{ij}\\ 
    0  &
    \end{pmatrix}
\end{align}
so that it is easy to compute
\begin{align} \begin{split}
    \left( ^E\Gamma  \right)^0_{00}=0, \quad \left( ^E\Gamma  \right)^0_{i0}=0, \quad \left( ^E\Gamma  \right)^0_{ij}=a\dot{a}\overline{g}_{ij} \\
    \left( ^E\Gamma  \right)^k_{00}=0, \quad \left( ^E\Gamma  \right)^k_{i0}=\frac{\dot{a}}{a}\delta_i^k, \quad \left( ^E\Gamma  \right)^k_{ij}=\overline{\Gamma}^k_{ij} \end{split}
\end{align}
where $\overline{\Gamma}^{k}_{ij}$ are the Christoffel symbols of $\overline{g}_{ij}$. From this we can directly compute the components of the Ricci tensor to obtain
\begin{align} \label{RicciComponentsRWMetric}\begin{split}
    R^E_{00}&=-3\frac{\ddot{a}}{a} \\
    R^E_{ij}&=\left( a\ddot{a} + 2\dot{a}^2 +2K  \right)\overline{g}_{ij} \\
    R^E_{0 i}&=0
\end{split}\end{align}
where the fact that for a $n$-dimensional spaceform $\overline{R}_{ij}=K(n-1)\overline{g}_{ij}$ with $K=-1$ for hyperbolic space, $K=0$ for Euclidean space, and $K=1$ for the sphere was used. This also gives us the value of the scalar curvature 
\begin{align}
    R^E=6\left( \frac{\ddot{a}}{a} + \frac{\dot{a}^2}{a^2} + \frac{K}{a^2}  \right).
\end{align}
As a check, these formulas lead to the same main equations (5.2.14) and (5.2.15) in \cite{Wald:1984rg} (the equations are given for $K=-1, 0, 1$ but the calculations, where $R^E_{\mu \nu}$ is an intermediate step are only done for the $K=0$ case).  

Next, by homogeneity and isotropy we only need to consider functions $F=F(t)$. We need to look at $\nabla^E_\mu \nabla^E_\nu F$. Recall, if given some one form $\omega_\mu$ and some metric $g$ with covariant derivative $\nabla$ one has
\begin{align}
    \nabla_\mu \omega_\nu = \partial_\mu \omega_\nu - \Gamma_{\mu \nu}^\rho \omega_\rho
\end{align}
but now if $F=F(t)$ then $\nabla^E_\mu F=F_\mu$ is a $1$-form with components
\begin{align}
    F_0=\dot{F}, \quad F_i = 0.
\end{align}
Therefore we get
\begin{align}
    \begin{split}
        \nabla^E_0 \nabla^E_0 F &= \partial_0 F_0 - \left(^E\Gamma\right)^0_{00} F_0 =\ddot{F} \\
        \nabla^E_i \nabla^E_0 F &= \partial_i F_0 - \left(^E\Gamma\right)^0_{i0} F_0 =0 \\
        \nabla^E_i \nabla^E_j F &= \partial_i F_j - \left(^E\Gamma\right)^0_{ij} F_0 = -a\dot{a} \dot{F} \overline{g}_{ij}
    \end{split}
\end{align}
for the components. Next we also get
\begin{align*} \begin{split}
    \square^E F = (g^{E})^{\mu \nu} \nabla^E_\mu \nabla^E_\nu F &= (g^E)^{00} \nabla^E_0 \nabla^E_0 F + (g^E)^{ij} \nabla^E_i \nabla^E_j F \\
    &=-\ddot{F} + a^{-2} \overline{g}^{ij}\left( -a\dot{a} \dot{F} \overline{g}_{ij} \right) \\
    &= -\ddot{F} -3 \frac{\dot{a}}{a}\dot{F}  \end{split}
\end{align*}
so to summarize
\begin{align} \label{DAlembertianRWMetric}
    \square^E F = -\ddot{F} - 3 \frac{\dot{a}}{a}\dot{F}
\end{align}
for the D'Almbertian.

Similarly, we want to compute all of the same quantities for the string metric 
\begin{align}
    g = e^{2\Phi} g^E = -e^{2\Phi} dt^2 + e^{2\Phi} a^2 d\Sigma^2
\end{align}
by \eqref{RelationEinsteinToStringMetric}. Due to homogeneity and isotropy we require $\Phi=\Phi(t)$.  Recall the general formula for Christoffel symbols of $g$ is
\begin{align}
    \Gamma^\gamma_{\alpha \beta}=\frac{1}{2}g^{\gamma \rho}\left( \partial_\alpha g_{\beta \rho} + \partial_\beta g_{\alpha \rho} - \partial_\rho g_{\alpha \beta}   \right).
\end{align} 
To make computations easier we write
\begin{align}
    g_{\mu \nu} =\begin{pmatrix}
     -e^{2\Phi} & 0 & 0  & 0 \\
     0 \\ 
    0 & & e^{2\Phi}a^2 \overline{g}_{ij}\\ 
    0  &
    \end{pmatrix}, \quad \quad  g^{\mu \nu} =\begin{pmatrix}
     -e^{-2\Phi} & 0 & 0  & 0 \\
     0 \\ 
    0 & & e^{-2\Phi} a^{-2} \overline{g}^{ij}\\ 
    0  &
    \end{pmatrix}
\end{align}

so we compute
\begin{align} \begin{split}
    \Gamma^0_{00}&=\frac{1}{2}g^{0 \rho}\left( \partial_0 g_{0 \rho} + \partial_0 g_{0 \rho} - \partial_\rho g_{00}   \right) = \frac{1}{2}g^{0 0}\left( \partial_0 g_{0 0} + \partial_0 g_{0 0} - \partial_0 g_{00}   \right)\\ &= \frac{1}{2}(-e^{-2\Phi})\left(-\partial_0 \left( e^{2\Phi}\right) \right) = \dot{\Phi} \\
    \Gamma^0_{i0}&=\frac{1}{2}g^{0 \rho}\left( \partial_i g_{0 \rho} + \partial_0 g_{i \rho} - \partial_\rho g_{i 0}   \right)=\frac{1}{2}g^{0 0}\left( \partial_i g_{0 0} + \partial_0 g_{i 0} - \partial_0 g_{i 0}   \right)= 0 \\
    \Gamma^0_{ij}&=\frac{1}{2}g^{0 \rho}\left( \partial_i g_{j \rho} + \partial_j g_{i \rho} - \partial_\rho g_{ij}   \right)= \frac{1}{2}g^{0 0}\left( \partial_i g_{j 0} + \partial_j g_{i 0} - \partial_0 g_{ij}   \right) \\
    &= \frac{1}{2}(-e^{-2\Phi})\left(-\partial_0 \left(  e^{2\Phi}a^2 \overline{g}_{ij} 
  \right) \right) = \frac{1}{2}e^{-2\Phi}\left(2a\dot{a} e^{2\Phi} + 2\dot{\Phi}e^{2\Phi} a^2  \right)\overline{g}_{ij}  \\
  &=\left(  a\dot{a} + a^2 \dot{\Phi}  \right)\overline{g}_{ij} \\
  \Gamma^k_{00} & = \frac{1}{2}g^{k \rho}\left( \partial_0 g_{0 \rho} + \partial_0 g_{0 \rho} - \partial_\rho g_{00}   \right) = \frac{1}{2}g^{k m}\left( \partial_0 g_{0 m} + \partial_0 g_{0 m} - \partial_m g_{00}   \right)=0 \\
  \Gamma^{k}_{i0}&= \frac{1}{2}g^{k \rho}\left( \partial_i g_{0 \rho} + \partial_0 g_{i \rho} - \partial_\rho g_{i 0}   \right)= \frac{1}{2}g^{k m}\left( \partial_i g_{0 m} + \partial_0 g_{i m} - \partial_m g_{i 0}   \right) \\
  &= \frac{1}{2}(e^{-2\Phi} a^{-2} \overline{g}^{km}) \partial_0 \left( e^{2\Phi} a^2 \overline{g}_{im} \right) = \frac{1}{2}e^{-2\Phi}a^{-2}\left( 2a\dot{a} e^{2\Phi} + 2\dot{\Phi}e^{2\Phi} a^2  \right) \overline{g}^{km}\overline{g}_{im} \\
  & = \left(  \frac{\dot{a}}{a} + \dot{\Phi}  \right) \delta^k_i \\
  \Gamma^{k}_{ij}&= \frac{1}{2}g^{k \rho}\left( \partial_i g_{j \rho} + \partial_j g_{i \rho} - \partial_\rho g_{ij}   \right)= \frac{1}{2}g^{k m}\left( \partial_i g_{j m} + \partial_j g_{i m} - \partial_m g_{ij}   \right) \\
  &= \frac{1}{2}\overline{g}^{k m}\left( \partial_i \overline{g}_{j m} + \partial_j \overline{g}_{i m} - \partial_m \overline{g}_{ij}   \right) = \overline{\Gamma}^k_{ij}.
  \end{split}
\end{align}

Next as before we have to compute $\nabla_\mu \nabla_\nu F$ for $F=F(t)$. In that case we have $\nabla_\mu F = F_\mu$ with $F_0 =\dot{F}$ and $F_i=0$. Therefore
\begin{align} \label{StringMetricHessian} \begin{split}
    \nabla_0 \nabla_0 F &= \partial_0 F_0 - \Gamma_{00}^0 F_0 = \ddot{F} - \dot{\Phi}\dot{F} \\
    \nabla_i \nabla_0 F &= \partial_i F_0 - \Gamma^0_{i0} F_0 = 0 \\
    \nabla_i \nabla_j F &= \partial_i F_j - \Gamma_{ij}^0 F_0 = -\dot{F}\left(  a\dot{a} + a^2 \dot{\Phi}  \right)\overline{g}_{ij} \end{split}
\end{align}
which also gives
\begin{align}
    \begin{split}
        \square F &= g^{\mu \nu} \nabla_\mu \nabla_\nu F = g^{00}\nabla_0 \nabla_0 F + g^{ij}\nabla_i \nabla_j F \\
        &=(-e^{-2\Phi})\left( \ddot{F} - \dot{\Phi} \dot{F} \right) + (e^{-2\Phi}a^{-2}\overline{g}^{ij}) \left(  -\dot{F}\left(  a\dot{a} + a^2 \dot{\Phi}  \right)\overline{g}_{ij} \right) \\
    &= e^{-2\Phi}\left( -\ddot{F} + \dot{\Phi} \dot{F} - \dot{F}\left( a\dot{a} + a^2 \dot{\Phi}\right)\overline{g}_{ij} a^{-2} \overline{g}_{ij}  \right)\\
    &= e^{-2\Phi} \left( -\ddot{F} + \dot{\Phi}\dot{F}-3\dot{F} \left( \frac{\dot{a}}{a} + \dot{\Phi}  \right)  \right)    \\
        & =e^{-2\Phi}\left( -\ddot{F} - 3\frac{\dot{a}}{a}\dot{F} - 2\dot{\Phi} \dot{F}  \right)
    \end{split}
\end{align}
for the D'Alembertian. To summarize
\begin{align}\label{DAlembertianStringMetric}
     \square F = e^{-2\Phi}\left( -\ddot{F} - 3\frac{\dot{a}}{a}\dot{F} - 2\dot{\Phi} \dot{F}  \right)
\end{align}
which can also be written as
\begin{align}\label{RelationBetweenDAlembertians}
     \square F = e^{-2\Phi}\left( \square^E F - 2\dot{\Phi} \dot{F}  \right)
\end{align}
by \eqref{DAlembertianRWMetric} for $F=F(t)$. 

Next, we need to compute the components of $R_{\alpha \beta}$ and $R$. We could do this directly from the Christoffel symbols. However, since we have \eqref{RicciComponentsRWMetric} and the two metrics are related by a conformal change, we can use the well known formula for the Ricci curvature and scalar curvature under a confromal change. The formulas say that if 
\begin{align}
    \widetilde{g}_{\mu \nu} = \Omega^2 \widehat{g}_{\mu \nu}
\end{align}
then 
\begin{align} \begin{split}
    \widetilde{R}_{\mu \nu} &= \widehat{R}_{\mu \nu}-(D-2)\widehat{\nabla}_\mu \widehat{\nabla}_\nu \ln\Omega - \widehat{g}_{\mu \nu} \widehat{g}^{\alpha \beta} \widehat{\nabla}_\alpha \widehat{\nabla}_\beta \ln \Omega + (D-2) (\widehat{\nabla}_\mu \ln \Omega)(\widehat{\nabla}_\nu \ln \Omega)\\
    & \quad - (D-2) \widehat{g}_{\mu \nu} \widehat{g}^{\alpha \beta} (\widehat{\nabla}_\alpha \ln \Omega)(\widehat{\nabla}_\beta \ln \Omega )
  \end{split}
\end{align}
and
\begin{align}
    \widetilde{R} = \Omega^{-2} \left( \widehat{R} -2 (D-1) \widehat{g}^{\alpha \beta} \widehat{\nabla}_\alpha \widehat{\nabla}_\beta \ln \Omega -(D-2)(D-1)\widehat{g}^{\alpha\beta}(\widehat{\nabla}_\alpha \ln \Omega)(\widehat{\nabla}_\beta \ln \Omega)  \right)
\end{align}
and so applying this to $D=4$ with $g_{\mu \nu}=e^{2\Phi} g^{E}_{\mu \nu}$ we have $\Omega=e^\Phi$ so $\ln(\Omega)=\Phi$ which gives us
\begin{align}
    R_{\mu \nu} &=R^E_{\mu \nu} -2 \nabla^E_\mu \nabla^E_\nu \Phi -g^E_{\mu \nu} (\square^E \Phi) + 2 \nabla^E_\mu \Phi \nabla^E_\nu \Phi - 2 g^E_{\mu \nu}  ((g^E)^{\alpha \beta} \nabla^E_\alpha \Phi \nabla^E_\beta \Phi) \\
    R & = e^{-2\Phi} \left( R^E - 6 (\square^E \Phi) - 6 ((g^E)^{\alpha \beta} \nabla^E_\alpha \Phi \nabla^E_\beta \Phi) \right)
\end{align}
In the case $\Phi=\Phi(t)$ we have
\begin{align}
    (g^E)^{\alpha \beta} \nabla^E_\alpha \Phi \nabla^E_\beta \Phi = -\dot{\Phi}^2
\end{align}
so we have
\begin{align} \label{RicciStringMetricR00} \begin{split}
    R_{00} &=R^E_{00} -2 \nabla^E_0 \nabla^E_0 \Phi -g^E_{00} (\square^E \Phi) + 2 \nabla^E_0 \Phi \nabla^E_0 \Phi - 2 g^E_{00}  (-\dot{\Phi}^2) \\
    &= -3\frac{\ddot{a}}{a}  -2 \ddot{\Phi} -(-1)(\square^E \Phi) + 2 \dot{\Phi}^2 -2 (-1)(-\dot{\Phi}^2) \\
    &= -3\frac{\ddot{a}}{a}  -2 \ddot{\Phi} +(\square^E \Phi) \end{split}
\end{align}

\begin{align}\label{RicciStringMetricRii} \begin{split}
   R_{ij}&= R^E_{ij} -2 \nabla^E_i \nabla^E_j \Phi -g^E_{ij} (\square^E \Phi) + 2 \nabla^E_i \Phi \nabla^E_j \Phi - 2 g^E_{ij}  (-\dot{\Phi}^2) \\
   &=\left( a\ddot{a} + 2\dot{a}^2 +2K  \right)\overline{g}_{ij} - 2\left( -a\dot{a} \dot{\Phi} \overline{g}_{ij} \right) -(a^2 \overline{g}_{ij})(\square^E \Phi)+0  + 2\dot{\Phi}^2 a^2 \overline{g}_{ij} \\
   &=\left( a\ddot{a} + 2\dot{a}^2 +2K +2a\dot{a} \dot{\Phi} -a^2 (\square^E \Phi )  +2 \dot{\Phi}^2 a^2  \right) \overline{g}_{ij} \end{split}
\end{align}

\begin{align}
    \begin{split}
     R_{0i} =R^E_{0i} -2 \nabla^E_0 \nabla^E_i \Phi -g^E_{0 i} (\square^E \Phi) + 2 \nabla^E_0 \Phi \nabla^E_i \Phi - 2 g^E_{0 i}  (-\dot{\Phi}^2) = 0
    \end{split}
\end{align}

\begin{align} \label{StringMetricScalarCurvature} \begin{split}
   R & = e^{-2\Phi} \left( R^E - 6 (\square^E \Phi) - 6 (-\dot{\Phi}^2) \right) \\
   &=e^{-2\Phi} \left[ 6\left( \frac{\ddot{a}}{a} + \frac{\dot{a}^2}{a^2} + \frac{K}{a^2}  \right) - 6(\square^E \Phi) + 6 \dot{\Phi^2}  \right] \end{split}
\end{align}

\subsection{Solving $\beta^B_{\mu \nu}=0$}

To obtain that $\beta^B_{\mu \nu}=0$ to leading order we need to solve the equation
\begin{align}
    \nabla^\rho \left( e^{-2\Phi} H_{\mu \nu \rho} \right)= 0
\end{align}
We can solve this by making the duality transformation 
\begin{align}\label{FormulaH}
    H_{\mu \nu \rho} = e^{2\Phi} \epsilon_{\lambda \mu \nu \rho} \nabla^\lambda b
\end{align}
where $\epsilon_{\lambda \mu \nu \rho}$ is a totally antisymmetric tensor compatible with the metric, meaning it satisfies $\nabla_\alpha \epsilon_{\beta \mu \nu \rho}=0 $, and $b$ is the axion field. To see that this solves the equation we compute
\begin{align}
    \nabla_\alpha \left( e^{-2\Phi} H_{\mu \nu \rho} \right) = \nabla_\alpha \left( \epsilon_{\lambda \mu \nu \rho} \nabla^\lambda b \right) = \left(\nabla_\alpha  \epsilon_{\lambda \mu \nu \rho}\right) \nabla^\lambda b + \epsilon_{\lambda \mu \nu \rho} \nabla_\alpha \nabla^\lambda b = \epsilon_{\lambda \mu \nu \rho} \nabla_\alpha \nabla^\lambda b
\end{align}
which therefore gives
\begin{align}
     \nabla^\rho \left( e^{-2\Phi} H_{\mu \nu \rho} \right) = \epsilon_{\lambda \mu \nu \rho} \nabla^\rho \nabla^\lambda b = \epsilon_{\lambda \mu \nu \rho} \nabla^\lambda \nabla^\rho b = -\epsilon_{\rho \mu \nu \lambda} \nabla^\lambda \nabla^\rho b =-\epsilon_{\lambda \mu \nu \rho} \nabla^\rho \nabla^\lambda b =0
\end{align}
where we used the symmetry $\nabla^\lambda \nabla^\rho b = \nabla^\rho \nabla^\lambda b$, the antisummetry of $\epsilon_{\lambda \mu \nu \rho}$, and the fact we can relabel contracted dummy indices. So this does indeed give a solution. However, the axion field $b$ is not free because $H_{\mu \nu \rho}$ has to satisfy the Jacobi identity
\begin{align}
    \nabla_{\left[ \alpha \right.} H_{\left. \mu \nu \rho   \right]} = 0.
\end{align}
We can compute 
\begin{align}
    \nabla_\alpha H_{\mu \nu \rho} = \nabla_\alpha \left( e^{2\Phi} \epsilon_{\lambda \mu \nu \rho} \nabla^\lambda b  \right) = \epsilon_{\lambda \mu \nu \rho}e^{2\Phi}( \nabla_\alpha \nabla^\lambda b  + 2\nabla_\alpha \Phi\nabla^\lambda b ).
\end{align}
Since $\nabla_{\left[ \alpha \right.} H_{\left. \mu \nu \rho   \right]}$ is antisymmetric, we only have to compute a single component for a distinct collection of indices. A direct computation then gives
\begin{align}
    \nabla_{\left[ 0 \right.} H_{\left. 123   \right]} = \frac{1}{4!} \sum_{\pi} \text{sgn}(\pi) \nabla_{\pi(0)} H_{ \pi(1) \pi(2) \pi(3)  } = \frac{1}{4}\epsilon_{0123}e^{2\Phi}\left( \square b + 2\nabla_\alpha \Phi\nabla^\alpha b  \right)
\end{align}
so that $b$ must satisfy 
\begin{align}
    \square b + 2\nabla_\alpha \Phi\nabla^\alpha b = 0.
\end{align}
The reason one gets this positive coefficient is because $\text{sgn}(\pi)$ and the antisymmetric tensor always give the same sign when permuting the indices. This equation matches the equation for the axion obtained in \cite{antoniadis1988cosmological} so this is a good sign. 

Now, by homogeneity and isotropy we require that $b=b(t)$ and $\Phi=\Phi(t)$. Therefore the above equation becomes 
\begin{align}
    e^{-2\Phi}\left( -\ddot{b} - 3\frac{\dot{a}}{a}\dot{b} - 2\dot{\Phi} \dot{b}  \right) - 2e^{-2\Phi}\dot{\Phi}\dot{b} = 0
\end{align}
or 
\begin{align}
    \ddot{b}+ \left(  3\frac{\dot{a}}{a} +4\dot{\Phi} \right)\dot{b} = 0
\end{align}
which is easily solved using the method of integrating factors to obtain
\begin{align}\label{Formulabdot}
    \dot{b}=\frac{\beta}{a^3 e^{4\Phi}}
\end{align}
where $\beta$ is some constant (potentially complex as we will see later).

Also by homogeneity and isotropy, all of our equations end up depending only on $t$ and can be done at a conveniently chosen point $(t, p)$ with $p\in \Sigma$. We can pick $(t, p)$ to be a point such that $\overline{g}_{ij}=\delta_{ij}$ so that
\begin{align} \label{DiagonalMetrics}
    \begin{split}
        g^E_{\mu \nu}&=\text{diag}(-1, a^2, a^2, a^2), \quad (g^E)^{\mu \nu}= \text{diag}(-1, a^{-2}, a^{-2}, a^{-2}) \\
        g_{\mu \nu} &= \text{diag}(-e^{2\Phi}, e^{2\Phi}a^2, e^{2\Phi}a^2, e^{2\Phi}a^2 ), \;  g^{\mu \nu} = \text{diag}(-e^{-2\Phi}, e^{-2\Phi}a^{-2}, e^{-2\Phi}a^{-2}, e^{-2\Phi}a^{-2} )
    \end{split}
\end{align}
which can be done by either choosing geodesic normal coodrinates in the space slice around $p$ or by noting the Robertson-Walker metrics can always be put in diagonal form in a (local or global, depending on $K$) coordinate patch, see (5.1.11) in \cite{Wald:1984rg}.

Now, since $\epsilon_{\lambda \mu \nu \rho}$ is totally antisymmetric and compatible with the metric, it is (up to a constant multiple) the volume form, so 
\begin{align}
    \epsilon_{\lambda \mu \nu \rho} = \sqrt{-g}\tilde{\epsilon}_{\lambda \mu \nu \rho}
\end{align}
where $\tilde{\epsilon}_{\lambda \mu \nu \rho}$ is the Levi-Civita symbol. But at our chosen point $\sqrt{-g}=e^{4\Phi}a^3$. Moreover, the only non-zero component of $\nabla^\lambda b$ is $\nabla^0 b = -e^{-2\Phi} \dot{b}$. But then we get
\begin{align}
    H_{\mu \nu \rho} = e^{2\Phi} \epsilon_{\lambda \mu \nu \rho} \nabla^\lambda b = e^{2\Phi}(e^{4\Phi}a^3 \tilde{\epsilon}_{0 \mu \nu \rho } )(-e^{-2\Phi})\left(  \frac{\beta}{a^3 e^{4\Phi}}  \right) = -\beta \tilde{\epsilon}_{0 \mu \nu \rho }
\end{align}
for the components. 

Next, we need to compute
\begin{align}
    H_{\mu \rho \sigma}{H_{\nu}}^{\rho \sigma}, \quad H^2=H_{\mu \nu \rho}H^{\mu \nu \rho}
\end{align}
since we will need these appear in $\beta^g_{\mu \nu}$ and $\beta^\Phi$. Since the metric is diagonal
\begin{align}
    H_{\mu \rho \sigma}{H_{\nu}}^{\rho \sigma} = H_{\mu \rho \sigma} H_{\nu \alpha \beta} g^{\alpha \rho} g^{\beta \sigma} = \sum_{\rho, \sigma} H_{\mu \rho \sigma} H_{\nu \rho \sigma} g^{\rho \rho} g^{\sigma \sigma}
\end{align}
but now if either $\mu$ or $\nu$ is $0$ then there is a repeated index in $\tilde{\epsilon}_{0 \mu \nu \rho}$ and so the term vanishes. So we have
\begin{align}\label{H00}
     H_{0 \rho \sigma}{H_{\nu}}^{\rho \sigma} =  H_{\mu \rho \sigma}{H_{0}}^{\rho \sigma}=0.
\end{align}
Next, let us look at the case of $H_{i \rho \sigma}{H_{j}}^{\rho \sigma}$ for $i\neq j$. We have
\begin{align}
    H_{i \rho \sigma}{H_{j}}^{\rho \sigma}  = \sum_{\rho, \sigma} H_{i \rho \sigma} H_{j \rho \sigma} g^{\rho \rho} g^{\sigma \sigma} = \sum_{\rho, \sigma} \beta^2 \tilde{\epsilon}_{0i \rho \sigma} \tilde{\epsilon}_{0j \rho \sigma} g^{\rho \rho} g^{\sigma \sigma}
\end{align}
but for $\tilde{\epsilon}_{0i \rho \sigma}$ to be non-zero, all of the indices have to be distinct, meaning either $\rho$ or $\sigma$ equals $j$, But then necessarily $\tilde{\epsilon}_{0j \rho \sigma}$ contains a repeated index and is thus $0$. Hence
\begin{align}
    H_{i \rho \sigma}{H_{j}}^{\rho \sigma} =0
\end{align}
for $i\neq j$. Finally we consider the case where $\mu=\nu=i$. It is easiest to do this for a particular value to see what is happening, so we set $i=1$. We have
\begin{align}\label{Hii}\begin{split}
    H_{1 \rho \sigma}{H_{1}}^{\rho \sigma}  &= \sum_{\rho, \sigma} H_{1 \rho \sigma} H_{1 \rho \sigma} g^{\rho \rho} g^{\sigma \sigma} = \sum_{\rho, \sigma} \beta^2 \tilde{\epsilon}_{01 \rho \sigma} \tilde{\epsilon}_{01 \rho \sigma} g^{\rho \rho} g^{\sigma \sigma} \\
    &=\beta^2 \tilde{\epsilon}_{01 23} \tilde{\epsilon}_{01 23} g^{22} g^{33} + \beta^2 \tilde{\epsilon}_{01 32} \tilde{\epsilon}_{01 32} g^{33} g^{22} \\
    &= \frac{2\beta^2}{a^4 e^{4\Phi}}\end{split}
\end{align}
are the only non-zero components. 

Next, let us look at $H^2$. The only non-zero terms are the ones with all distinct indices none of which is $0$. Let us look at
\begin{align}
    H_{123}H^{123}=H_{123} H_{\alpha \beta \gamma} g^{1\alpha}g^{2\beta}g^{3\gamma} =(H_{123})^2 g^{11}g^{22} g^{33} = \frac{\beta^2}{a^6 e^{6\Phi}}
\end{align}
since the metric is diagonal. Since there are $6$ ways of permuting the $3$ nonzero indices we get
\begin{align} \label{H}
    H^2 = \frac{6 \beta^2}{a^6 e^{6\Phi}}.
\end{align}

\subsection{Solving $\beta^g_{\mu \nu}=0$} Now that we have done the preparatory work in the previous two sections, this section is quite easy. Same as in the previous section, we will pick a point such that the space metric is $\overline{g}_{ij}=\delta_{ij}$ and do all of our calculations at that point. In this case, we only have to look at $\beta^g_{00}$ and $\beta^g_{ii}$ since the other components vanish. Up to leading order we have
\begin{align}
    \begin{split}
         \frac{\beta_{00}^g}{\ell_s^2} &= R_{00} + 2\nabla_0 \nabla_0 \Phi -\frac{1}{4}H_{0 \rho \sigma}{H_{0}}^{\rho \sigma} \\
         &=\left(-3\frac{\ddot{a}}{a}  -2 \ddot{\Phi} +(\square^E \Phi) \right) + 2 \left( \ddot{\Phi} -\dot{\Phi}^2  \right) - 0 \\
         &= -3\frac{\ddot{a}}{a}  +(\square^E \Phi) -2\dot{\Phi}^2
    \end{split}
\end{align}
and 
\begin{align}\begin{split}
     \frac{\beta_{ii}^g}{\ell_s^2} &= R_{ii} + 2\nabla_i \nabla_i \Phi -\frac{1}{4}H_{i \rho \sigma}{H_{i}}^{\rho \sigma} \\
     &=\left( a\ddot{a} + 2\dot{a}^2 +2K +2a\dot{a} \dot{\Phi} -a^2 (\square^E \Phi )  +2 \dot{\Phi}^2 a^2  \right) \\&\quad  + 2 \left( -\dot{\Phi}(a\dot{a} + a^2 \dot{\Phi}) \right) - \frac{1}{4}\left( \frac{2\beta^2}{a^4 e^{4\Phi}} \right) \\
     &= a\ddot{a} + 2\dot{a}^2 +2K +2a\dot{a} \dot{\Phi} -a^2 (\square^E \Phi )  +2 \dot{\Phi}^2 a^2   -2a\dot{a}\dot{\Phi} -2a^2 \dot{\Phi}^2 -  \frac{\beta^2}{2a^4 e^{4\Phi}} \\
     &= a\ddot{a} + 2\dot{a}^2 +2K -a^2 (\square^E \Phi )  -  \frac{\beta^2}{2a^4 e^{4\Phi}}\end{split}
\end{align}
by \eqref{StringMetricHessian}, \eqref{RicciStringMetricR00}, \eqref{RicciStringMetricRii}, \eqref{H00}, and \eqref{Hii}, so if we can solve the system of equations
\begin{align}
 \begin{split}
        -3\frac{\ddot{a}}{a}  +(\square^E \Phi) -2\dot{\Phi}^2 = 0 \\
        a\ddot{a} + 2\dot{a}^2 +2K -a^2 (\square^E \Phi )  -  \frac{\beta^2}{2a^4 e^{4\Phi}}=0
    \end{split}
\end{align}
then $\beta^g_{\mu \nu}$ vanish to leading order. Moreover, if we multiply the first equation by $a^2$ and add it to the second we get the equivalent system 
\begin{align}
    \label{MainSystemOfEquations} \begin{split}
        -3\frac{\ddot{a}}{a}  +(\square^E \Phi) -2\dot{\Phi}^2 = 0 \\
        -2a\ddot{a} + 2\dot{a}^2 +2K  -  \frac{\beta^2}{2a^4 e^{4\Phi}} - 2\dot{\Phi}^2 a^2=0
    \end{split}
\end{align}
and since this is a highly non-linear second order system of ODEs we will solve it numerically in Section \ref{SecAnalysis}. 

\subsection{Analyzing $\beta^\Phi$} For $D=4$ the formula \eqref{BetaFunctionPhi} gives
\begin{align}
    \beta^\Phi=-22 +\frac{3}{2}\ell_s^2 \left[ 4\nabla^\mu \Phi \nabla_\mu \Phi -4 \square \Phi -R + \frac{1}{12}H^2\right] 
    \end{align}
to leading order. Let us simplify the term in the brackets. We have
\begin{align}
    4\nabla^\mu \Phi \nabla_\mu \Phi = 4 g^{\alpha \mu} \nabla_\alpha \Phi \nabla_\mu \Phi = -4e^{-2\Phi} \dot{\Phi}^2
\end{align}
\begin{align}
    -4\square \Phi = -4e^{-2\Phi}\left( \square^E \Phi - 2\dot{\Phi}^2   \right)=e^{-2\Phi} \left( -4\square^E \Phi + 8 \dot{\Phi}^2 \right)
\end{align}
where we used \eqref{RelationBetweenDAlembertians},
\begin{align}
    -R=e^{-2\Phi} \left[ -6\left( \frac{\ddot{a}}{a} + \frac{\dot{a}^2}{a^2} + \frac{K}{a^2}  \right) + 6(\square^E \Phi) - 6 \dot{\Phi^2}  \right]
\end{align}
from \eqref{StringMetricScalarCurvature}, and
\begin{align}
    \frac{1}{12}H^2 =\frac{1}{12}\frac{6 \beta^2}{a^6 e^{6\Phi}} = e^{-2\Phi}\frac{\beta^2}{2a^6 e^{4\Phi}}
\end{align}
from \eqref{H}. Since every term has a factor of $e^{-2\Phi}$ we can factor that out to get
\begin{Small}
\begin{align} \begin{split}
    \beta^\Phi&=-22+\frac{3}{2}\ell_s^2 e^{-2\Phi} \left[  -4\dot{\Phi}^2 -4(\square^E \Phi)+8\dot{\Phi}^2 -6\left( \frac{\ddot{a}}{a} + \frac{\dot{a}^2}{a^2} + \frac{K}{a^2}  \right) + 6(\square^E \Phi) - 6 \dot{\Phi^2} + \frac{\beta^2}{2a^6 e^{4\Phi}}\right] \\
    &=-22 +\frac{3}{2}\ell_s^2 e^{-2\Phi}\left[ -2\dot{\Phi}^2 +2 (\square^E \Phi) -6\frac{\ddot{a}}{a} -6 \frac{\dot{a}^2}{a^2} -6 \frac{K}{a^2} + \frac{\beta^2}{2a^6 e^{4\Phi}}  \right]. \end{split}
\end{align}
\end{Small}
Now, the idea is that if \eqref{MainSystemOfEquations} holds, then we can actually eliminate the second derivative terms and this only depends on $a, \dot{a}, \Phi, \dot{\Phi}$, as we will shot momentarily. Now, we have the well known fact that $\beta^g_{\mu \nu}= \beta^B_{\mu \nu}=0$ implies $\beta^\Phi=constant$ (see for example the paragraph after equation (6.1.8) in \cite{Kiritsis:2019npv}), so we only need $\beta^\Phi=0$ at one point in time. However, the constancy of $\beta^\Phi$ should also directly follow from \eqref{MainSystemOfEquations}. We check that this is indeed true, as this can be used as a double check of our work. If it is true, it is very unlikely we had made a mistake somewhere earlier in our calculations.

\begin{proposition}
    If $a(t)$ and $\Phi(t)$ satisfy the system \eqref{MainSystemOfEquations}, then the quantity
    \begin{align}
        \beta^\Phi = -22 +\frac{3}{2}\ell_s^2 e^{-2\Phi}\left[ -2\dot{\Phi}^2 +2 (\square^E \Phi) -6\frac{\ddot{a}}{a} -6 \frac{\dot{a}^2}{a^2} -6 \frac{K}{a^2} + \frac{\beta^2}{2a^6 e^{4\Phi}}  \right]
    \end{align}
is constant. 
\end{proposition}

\begin{proof}
    We need to check that if \eqref{MainSystemOfEquations} holds then 
    \begin{align}
        e^{-2\Phi}\left[ -2\dot{\Phi}^2 +2 (\square^E \Phi) -6\frac{\ddot{a}}{a} -6 \frac{\dot{a}^2}{a^2} -6 \frac{K}{a^2} + \frac{\beta^2}{2a^6 e^{4\Phi}}  \right]
    \end{align}
is constant. First notice that the first equation of \eqref{MainSystemOfEquations} can be written as 
\begin{align}
    -3\frac{\ddot{a}}{a}=2\dot{\Phi}^2 - \square^E \Phi
\end{align}
and also equivalently
\begin{align}
    \ddot{\Phi}=-3\frac{\ddot{a}}{a}-2\dot{\Phi}^2 - 3 \frac{\dot{a}}{a}\dot{\Phi}
\end{align}
so 
 \begin{align} \label{SimplifiedTermInBetaPhi} \begin{split}
        e^{-2\Phi}&\left[ -2\dot{\Phi}^2 +2 (\square^E \Phi) -6\frac{\ddot{a}}{a} -6 \frac{\dot{a}^2}{a^2} -6 \frac{K}{a^2} + \frac{\beta^2}{2a^6 e^{4\Phi}}  \right] \\
        &=e^{-2\Phi}\left[ -2\dot{\Phi}^2 +2 (\square^E \Phi)+ 4\dot{\Phi}^2 - 2(\square^E \Phi) -6 \frac{\dot{a}^2}{a^2} -6 \frac{K}{a^2} + \frac{\beta^2}{2a^6 e^{4\Phi}}  \right]\\
        &=e^{-2\Phi}\left[ 2\dot{\Phi}^2  -6 \frac{\dot{a}^2}{a^2} -6 \frac{K}{a^2} + \frac{\beta^2}{2a^6 e^{4\Phi}}  \right] \end{split}
    \end{align}
which notice does not depend on any second derivatives, as mentioned earlier. Let us first look at 
\begin{align}\begin{split}
    \frac{d}{dt}&\left[ 2\dot{\Phi}^2  -6 \frac{\dot{a}^2}{a^2} -6 \frac{K}{a^2} + \frac{\beta^2}{2a^6 e^{4\Phi}}  \right]\\ &=4\dot{\Phi} \ddot{\Phi} -12 \frac{\dot{a}\ddot{a}}{a^2} + 12 \frac{\dot{a}^3}{a^3} + 12K \frac{\dot{a}}{a^3} + \frac{\beta^2}{2}\left( -6a^{-7}\dot{a}e^{-4\Phi} - 4\dot{\Phi} e^{-4\Phi} a^{-6} \right) \\
    & = 4\dot{\Phi}\left( -3\frac{\ddot{a}}{a}-2\dot{\Phi}^2 - 3 \frac{\dot{a}}{a}\dot{\Phi}  \right) -12 \frac{\dot{a}\ddot{a}}{a^2} + 12 \frac{\dot{a}^3}{a^3} + 12K \frac{\dot{a}}{a^3} - \frac{3\beta^2 \dot{a}}{a^7 e^{4\Phi}} - \frac{2\beta^2 \dot{\Phi}}{a^6 e^{4\Phi}} \\
    & = -12 \dot{\Phi}\frac{\ddot{a}}{a} -8\dot{\Phi}^3 -12 \frac{\dot{a}}{a}\dot{\Phi}^2 - 12\frac{\dot{a}}{a} \frac{\ddot{a}}{a} + 12 \frac{\dot{a}^3}{a^3} + 12K \frac{\dot{a}}{a^3} - \frac{3\beta^2 \dot{a}}{a^7 e^{4\Phi}} - \frac{2\beta^2 \dot{\Phi}}{a^6 e^{4\Phi}} \end{split}
\end{align}
but we can rewrite the second equation of \eqref{MainSystemOfEquations} as 
\begin{align}
    \frac{\ddot{a}}{a}=\frac{\dot{a}^2}{a^2} + \frac{K}{a^2} - \frac{\beta^2}{4a^6 e^{4\Phi}} - \dot{\Phi}^2
\end{align}
which substituting in gives
\begin{align}
    \begin{split}
        \frac{d}{dt}&\left[ 2\dot{\Phi}^2  -6 \frac{\dot{a}^2}{a^2} -6 \frac{K}{a^2} + \frac{\beta^2}{2a^6 e^{4\Phi}}  \right]\\ 
        &=-12\dot{\Phi}\left( \frac{\dot{a}^2}{a^2} + \frac{K}{a^2} - \frac{\beta^2}{4a^6 e^{4\Phi}} - \dot{\Phi}^2  \right)-8\dot{\Phi}^3 - 12 \frac{\dot{a}}{a}\dot{\Phi}^2 \\
        &-12 \frac{\dot{a}}{a}\left(  \frac{\dot{a}^2}{a^2} + \frac{K}{a^2} - \frac{\beta^2}{4a^6 e^{4\Phi}} - \dot{\Phi}^2 \right) + 12 \frac{\dot{a}^3}{a^3} + 12K \frac{\dot{a}}{a^3} - \frac{3\beta^2 \dot{a}}{a^7 e^{4\Phi}} - \frac{2\beta^2 \dot{\Phi}}{a^6 e^{4\Phi}} \\
        &=-12\dot{\Phi}\frac{\dot{a}^2}{a^2} - 12\dot{\Phi}\frac{K}{a^2} + \frac{3\beta^2 \dot{\Phi}}{a^6 e^{4\Phi}} + 12\dot{\Phi}^3 -8\dot{\Phi}^3 - 12 \frac{\dot{a}}{a}\dot{\Phi}^2 \\
        &-12\frac{\dot{a}^3}{a^3} -12K\frac{\dot{a}}{a^3} + \frac{3\beta^2 \dot{a}}{a^7 e^{4\Phi}} + 12\frac{\dot{a}}{a}\dot{\Phi}^2  + 12 \frac{\dot{a}^3}{a^3} + 12K \frac{\dot{a}}{a^3} - \frac{3\beta^2 \dot{a}}{a^7 e^{4\Phi}} - \frac{2\beta^2 \dot{\Phi}}{a^6 e^{4\Phi}} \\
        &=-12\dot{\Phi}\frac{\dot{a}^2}{a^2} - 12\dot{\Phi}\frac{K}{a^2} + \frac{\beta^2 \dot{\Phi}}{a^6 e^{4\Phi}} + 4\dot{\Phi}^3
    \end{split}
\end{align}
but now we have
\begin{align}\begin{split}
    \frac{d}{dt}&\left(  e^{-2\Phi}\left[ 2\dot{\Phi}^2  -6 \frac{\dot{a}^2}{a^2} -6 \frac{K}{a^2} + \frac{\beta^2}{2a^6 e^{4\Phi}}  \right]   \right)\\ &= -2e^{-2\Phi}\dot{\Phi}\left[ 2\dot{\Phi}^2  -6 \frac{\dot{a}^2}{a^2} -6 \frac{K}{a^2} + \frac{\beta^2}{2a^6 e^{4\Phi}}  \right] + e^{-2\Phi} \left[ -12\dot{\Phi}\frac{\dot{a}^2}{a^2} - 12\dot{\Phi}\frac{K}{a^2} + \frac{\beta^2 \dot{\Phi}}{a^6 e^{4\Phi}} + 4\dot{\Phi}^3  \right] \\
    &=e^{-2\Phi} \left( -4\dot{\Phi}^3  +12 \frac{\dot{a}^2}{a^2} \dot{\Phi} +12 \frac{K}{a^2}\dot{\Phi} - \frac{\beta^2}{a^6 e^{4\Phi}}\dot{\Phi} -12\dot{\Phi}\frac{\dot{a}^2}{a^2} - 12\dot{\Phi}\frac{K}{a^2} + \frac{\beta^2 \dot{\Phi}}{a^6 e^{4\Phi}} + 4\dot{\Phi}^3  \right)\\
    &=0\end{split}
\end{align}
as expected. \end{proof}

Again, we didn't have to actually do this calculation, but the fact that the expected result comes out strongly suggests that all of our previous calculations are indeed correct. To summarize, if \eqref{MainSystemOfEquations} is satisfied then we can write
\begin{align}\label{FinalFormulaBetaPhi}
    \beta^\Phi = -22 +\frac{3}{2}\ell_s^2 e^{-2\Phi}\left[ 2\dot{\Phi}^2  -6 \frac{\dot{a}^2}{a^2} -6 \frac{K}{a^2} + \frac{\beta^2}{2a^6 e^{4\Phi}}  \right] 
\end{align}
which is determined purely by the initial conditions for our system of equations.

\section{Analysis of the Equations}\label{SecAnalysis}

Since our observed universe looks approximately static, we start by looking for static solutions to our equations. We have the following theorem:

\begin{theorem}
    If 
    \begin{align}
        \Phi(t)=\Phi_0, \quad a(t) =a_0, \quad K=-1, \quad \beta^2=-4a_0^4 e^{4\Phi_0}, \quad \ell_s^2 = \frac{22 }{6} a_0^2 e^{2\Phi_0}
    \end{align}
then all of the beta functions vanish, giving a $4$-dimensional background in which string theory is consistent. This gives a purely complex axion. We refer to the resulting spacetime as an anti-Einstein static universe.      
\end{theorem}

\begin{proof}
    If $\Phi(t)=\Phi_0$ and $a(t)=a_0$ then the first equation in \eqref{MainSystemOfEquations} reduces to $0=0$ and holds trivially. The second equation reduces to 
    \begin{align} \label{CriticalEquation}
        K= \frac{\beta^2}{4a_0^4 e^{4\Phi_0}}
    \end{align}
    or \begin{align}
        \beta^2 =4a_0^4 e^{4\Phi_0} K.
    \end{align}
We therefore have
\begin{align}\begin{split}
     \beta^\Phi &= -22 +\frac{3}{2}\ell_s^2 e^{-2\Phi_0}\left[  -6 \frac{K}{a_0^2} + \frac{\beta^2}{2a_0^6 e^{4\Phi_0}}  \right] \\
     &=-22 +\frac{3}{2}\frac{\ell_s^2 e^{-2\Phi_0}}{a_0^2}\left[  -6K + \frac{\beta^2}{2a_0^4 e^{4\Phi_0}}  \right] \\
     &=-22 +\frac{3}{2}\frac{\ell_s^2 e^{-2\Phi_0}}{a_0^2}\left(  -6K  +2K  \right) \\
     &=-22 +\frac{3}{2}\frac{\ell_s^2 e^{-2\Phi_0}}{a_0^2}\left(  -4K  \right)\end{split}
\end{align}
so to have $\beta^{\Phi}=0$ we need $K=-1$ after which some simple algebra gives
\begin{align}
    \ell_s^2 = \frac{22 }{6} a_0^2 e^{2\Phi_0}.
\end{align}
    Finally, letting $\dot{b}$ and $H_{\mu \nu \rho}$ be defined by \eqref{Formulabdot} and \eqref{FormulaH} ensures $\beta^B_{\mu \nu}=0$.
\end{proof}

Notice, this is also very interesting because using equations \eqref{RWDDotEquation} and \eqref{RWDotEquation} the anti-Einstein static universe can be interpreted as a universe filled with a perfect fluid having a negative energy density and positive pressure. Interestingly, in \cite{farnes2018unifying} a similar "dark fluid" was used as an adhoc hypothesis to explain dark matter. Here, such a dark fluid pops out of the equations on its own. Since our universe is close to static, computing galaxy rotation curves against the anti-Einstein static background to see if this reproduces observations is a natural future step which we hope to do sometime in the future.

We find that the equations are very sensitive to the initial conditions and tend to blow up either in the future or in the past if one picks initial conditions too far away from the static case. Since our universe appears quite close to being static, we will choose some initial conditions which are close to the initial conditions for the static case. However, the question is, what are the correct initial conditions? At least two of these are easy to set, but when it comes to the conditions for the dilaton field and the axion constant this is less clear. For now, we just play around with the initial conditions to show that the solutions have rich behavior in both the past and future. The hope is that we can find initial conditions which will match current models for the expansion parameter $a(t)$ for the universe, but this will require some kind of algorithmic search through the space of initial conditions restricted by the condition $\beta^\Phi=0$.

Recall, one defines the Hubble parameter by
\begin{align}
    h(t)=\frac{\dot{a}(t)}{a(t)}.
\end{align}
where we use a lowercase $h$ so as not to be confused with $H$ appearing in the beta functions. Notice an increasing $h$ implies a necessarily accelerating universe since
\begin{align}
    \dot{h} = \frac{\ddot{a}}{a} - \frac{\dot{a}^2}{a^2}
\end{align}
so that the first term has to be positive. We can always normalize and set $a(0)=1$. We will use the currently measured value of the Hubble parameter
\begin{align}
    h(0)=\frac{\dot{a}(0)}{a(0)}=\dot{a}(0) \approx 2 \times 10^{-18} s^{-1}.
\end{align}
for our $\dot{a}(0)$. Since we are looking at the evolution of the universe, It turns out to be easiest to measure time in billions of years. Since $1 \text{ billion years} \approx 3 \times 10^{16} s$ in these units we get
\begin{align}
    \dot{a}(0)\approx 6 \times  10^{-2} \text{(billion years)}^{-1}.
\end{align}
So in all of our graphs, time is measured in billions of years. We will let $t=0$ be current time.

Next, we need to rewrite the system \eqref{MainSystemOfEquations} in a manner more suitable for numerical integration. The first equation can be rewritten as
\begin{align}
    \ddot{\Phi}=-2\dot{\Phi}^2 - 3 \frac{\dot{a}}{a}\dot{\Phi}    -3\frac{\ddot{a}}{a}
\end{align}
and the second as
\begin{align}
     {\ddot{a}}=\frac{\dot{a}^2}{a} + \frac{K}{a} - \frac{\beta^2}{4a^5 e^{4\Phi}} - a\dot{\Phi}^2
\end{align}
and substituting this into the first equation gives the system
\begin{align} \label{NumericalSystemOfEquations} \begin{split}
    \ddot{\Phi}&=-2\dot{\Phi}^2 - 3 \frac{\dot{a}}{a}\dot{\Phi}    -3\left( \frac{\dot{a}^2}{a^2} + \frac{K}{a^2} - \frac{\beta^2}{4a^6 e^{4\Phi}} - \dot{\Phi}^2  \right) \\
    {\ddot{a}}&=\frac{\dot{a}^2}{a} + \frac{K}{a} - \frac{\beta^2}{4a^5 e^{4\Phi}} - a\dot{\Phi}^2\end{split}
\end{align}
which is easy to solve numerically by converting it to a first order system. We now give some examples of numerical solutions with initial conditions close to those of the anti-Einstein static universe. If we set $\Phi_0=0, a_0=1, K=-1$ then $\beta^2=-4$ so we will look at initial conditions close to the these values. Again, the idea here is to show the equations have enough interesting behavior to potentially match the curves for the Hubble parameter that is expected by the cosmologists. 

\subsection{Example 1} In this example we take the parameters to be
\begin{align}
    K=-1, \quad  \beta^2=-4.1, \quad a(0)=1, \quad  \dot{a}(0)=0.06, \quad \Phi(0)=-0.02, \quad \dot{\Phi}(0)=0
\end{align}
which upon solving \eqref{NumericalSystemOfEquations} gives the two graphs

\begin{figure}[H]
    \centering
    \subfloat[\centering Solution to the Past]{{\includegraphics[width=3in]{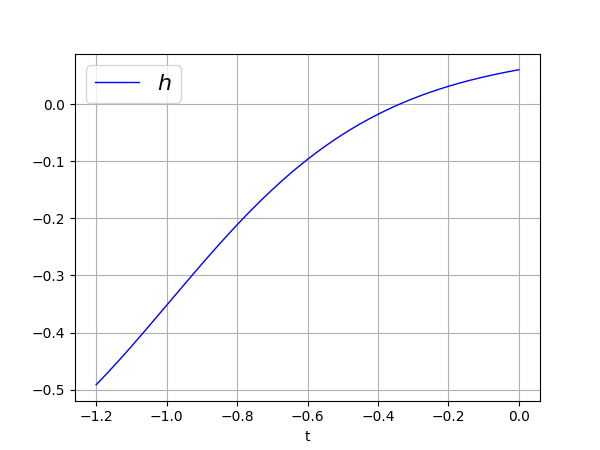} }}%
    \qquad
    \subfloat[\centering Solution to the Future]{{\includegraphics[width=3in]{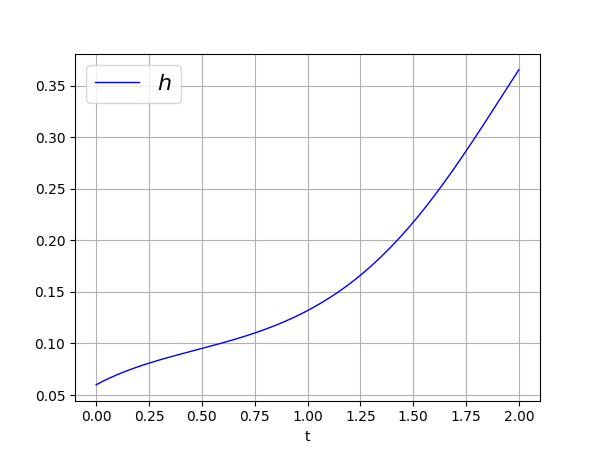} }}
\end{figure}
\noindent which we think is by far the most interesting example of the ones we will show. Notice in this case, the Hubble parameter is always increasing so $\ddot{a}>0$. Moreover, notice that $h<0$ in the past indicating a contraction, and there is a point where the sign changes so the contraction turns into expansion. Such a point could be thought of as a type of big bang. Certainly, going from contraction to expansion, is an interesting phenomenon worthy of some investigation. Also notice that the equations allow us to go before this "big bang". These conditions also give that $h$ will continue to increase into the future. Notice the time scale is not quite correct, since we are measuring $t$ in billions of years, this would say the big bang would have been about $400$ million years ago. However, the idea is that by playing around with the initial conditions we could move this point further backwards in time. For now, we are just interested in some general features of the solutions, and finding the actual correct initial conditions based on cosmological data is a project in and of itself. 

Moreover, the James Webb space telescope has discovered galaxies that seem to have formed way too soon after the expected date of the big bang. The period of contraction before our "big bang" (which we put in quotes as we are unsure if this point has the properties that cosmologists would find desirable in their model to fit being called a big bang) might give a possible explanation. 

\subsection{Example 2} For this one we use the initial conditions
\begin{align}
     K=-1, \quad  \beta^2=-4, \quad a(0)=1, \quad  \dot{a}(0)=0.06, \quad \Phi(0)=0, \quad \dot{\Phi}(0)=0
\end{align}
and we obtain
\begin{figure}[H]
    \centering
    \subfloat[\centering Solution to the Past]{{\includegraphics[width=3in]{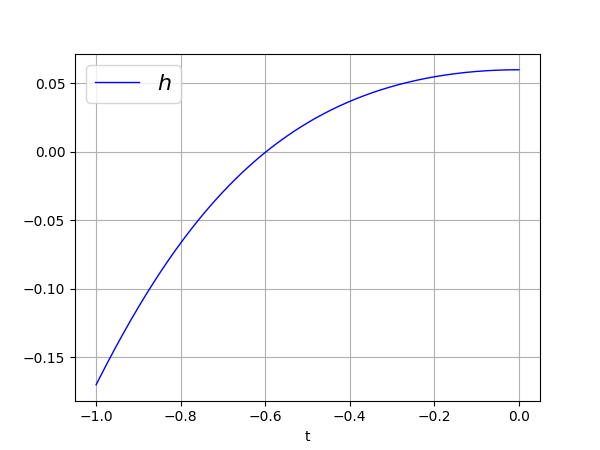} }}%
    \qquad
    \subfloat[\centering Solution to the Future]{{\includegraphics[width=3in]{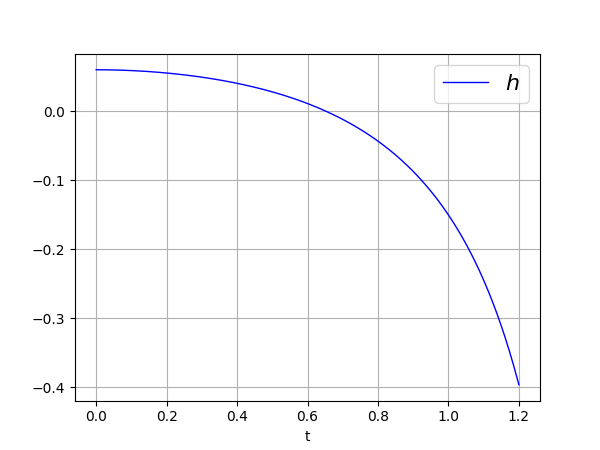} }}
\end{figure}
and notice that again there is an instant of time in the past where contraction turns into expansion, but not in the future the universe begins contracting again. So this is quite interesting because if the Hubble parameter is indeed approximately constant, by continuity there should be some initial conditions inbetween those of our two examples where the Hubble parameter is approximately constant for arbitrarily long periods of time. Again, what those conditions are requires a numerical search throguh the space of admissible initial conditions, constrained by $\beta^\Phi=0$. 

\subsection{Example 3} For this example use use the initial conditions
\begin{align}
     K=-1, \quad  \beta^2=-4, \quad a(0)=1, \quad  \dot{a}(0)=0.06, \quad \Phi(0)=0.1, \quad \dot{\Phi}(0)=0.06
\end{align}
which gives
\begin{figure}[H]
    \centering
    \subfloat[\centering Solution to the Past]{{\includegraphics[width=3in]{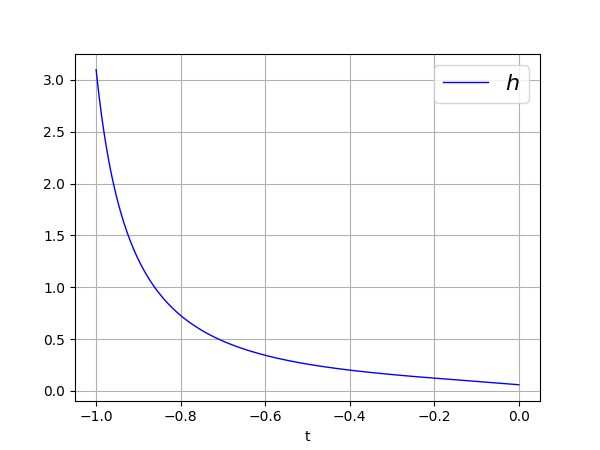} }}%
    \qquad
    \subfloat[\centering Solution to the Future]{{\includegraphics[width=3in]{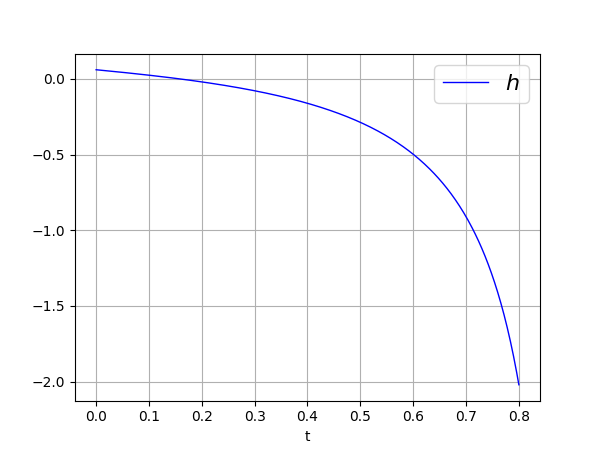} }}
\end{figure}
which notice has completely different behavior in the past than the other two examples, giving a Hubble parameter which blows up as we go into the past. This could more accurately be described as a period of cosmic inflation.

Overall, we hope that these examples are enough to suggest that the equations have enough interesting behavior to be able to fit the data of cosmologists. Whether that is true must be the subject of further investigation. Moreover, the equations currently do not contain any terms describing ordinary matter, so perhaps by including such terms by including massive fields in the Polyakov action the fit could be made even better.

\section{Conclusion and Future Research}

We started with the hypothesis that $D=4$ critical string theory is possible in a Robertson-Walker spacetime. This naturally led us to the idea of a purely complex axion field. A good way to judge if a hypothesis is good is how many problems it seems to solve. We introduced the above hypothesis to get rid of the problem of extra dimensions, but we seem to solve several additional problems for free. To recap:
\begin{enumerate}
    \item There is no need for extra dimensions, compactifications, branes, or Calabi-Yau manifolds. There is basically a unique static solution that gives a consistent $4$-dimensional string theory on cosmological scales which we call the anti-Einstein static universe.
    \item The anti-Einstein static universe has negative curvature. This implies it can be interpreted as being filled with a fluid of negative energy and positive pressure. In \cite{farnes2018unifying} a negative energy fluid was used as an adhoc assumption for resolving the dark matter problem. Assuming that analysis is correct, $D=4$ forces such a fluid upon us automatically, thus gives a way to resolve the dark matter problem.
    \item Small perturbations away from the static case give very interesting behaviors for the universe in both the past and the future. Some of these can be identified with inflation, some with something akin to a big bang. It should be possible to chose initial conditions to agree with cosmological observations and make predictions based on that which should be checkable. 
    \item One major problem in cosmology is the problem of singularities. The Hawking-Penrose singluarity theorems guarantee the existnce of singularities, but require certain energy conditions to hold \cite{Hawking:1970zqf}. These energy conditions do not hold in our static case, so the existence of singularities is no longer guaranteed. 
\end{enumerate}
so for the price of one hypothesis, it would seem we have solutions to four major problems. That is not bad. Next we discuss some other problems that could be solved.

\subsection{Connection to the Standard Model} Naturally, the goal is to eventually connect the theory to particle behavior. Since we have a unique static spacetime, we should be able to do any calculations against this background and somehow we need to connect this to the standard model. Ignoring the Higgs boson for the moment (which might be somehow connected to the dilaton) it would be foolish not to notice that there are $4$ bosons in the standard model, and now we have $4$ dimensions. Bosons are supposed to be interpreted as excitations of the coordinate fields on the string worldsheet. There are many ways of extracting exactly $4$ bosons from the four coordinates. One is to assume every string has to be excited in the time direction. This then leaves lowest order excitations in either $0, 1, 2,$ or $3$ of the spatial coordinates. Another possibility is to Wick rotate on the tangent space to get $4$-dimensional Euclidean space where all of the coordinates are equivalent, and then one can have lowest order excitations in $1, 2, 3$ or $4$ coordinates. Naturally, the next step is to include supersymmetry. Unbroken supersymmetry predicts that the bosons and corresponding fermions should have the same masses. Clearly, the supersymmetry has to be broken somehow, perhaps through an intercation with either the dilaton or axion. However, the fact we can start making connections to the four observed bosons is a step in the right direction.

\subsection{Connection to other theories of quantum gravity} If we draw an analogy to the standard model, there particles are quanta of fields. In string theory, particle properties are a result of quantizing fields on the worldsheet. But this then begs the question, what is the "string" field of which the strings are quanta? Geometry can be described in terms of holonomies which are loops in space. Perhaps we can consider some holonomy field, whatever that would mean, and strings could be quanta of such a field? This idea starts to look suspiciously like loop quantum gravity, and it would be extremely interesting if these competing theories turned out to be complementary in $D=4$. 

Also, in $D=4$ there are phenomena that do not occur in higher dimensions. In three spatial dimensions one can make non-trivial knots which is impossible in higher dimensions. Perhaps this can somehow cause supersymmetry breaking? Also, letting one's imagination run wild, one can imagine the background fields somehow being made of strings meshing together to form a coherent state, which is impossible in higher dimensions. 

\subsection{Future Research} We now comment on the next natural avenues of research. 
\begin{enumerate}
    \item Can initial conditions for \eqref{NumericalSystemOfEquations} be chosen such that they reproduce the Hubble parameter graphs believed to be true by the cosmologists?
    \item Can computation of galaxy rotation curves against the anti-Einstein static universe background reproduce the flattening of the rotation curves?
    \item We know what happens when space is homogenous and isotropic. What about the case where it is spherically symmetric? Will singularities, for example at the center of black holes still form? Or will the consistency of string theory prevent their formation?
\end{enumerate}
These seem to be the most natural questions to start with. Of course the other problems of connecting the theory to the standard model is also very interesting, but that might be something that is a bit harder to do. We hope to investigate all of these problems in the near future.

\bibliography{ref.bib}

 \footnotesize

  J.S.~Jaracz, \textsc{Department of Mathematics, Texas State University,
    San Marcos, TX 78666}\par\nopagebreak
  \textit{E-mail address} \texttt{jaracz@txstate.edu}

\end{document}